\documentclass[runningheads]{llncs}
\usepackage{float}
\usepackage{graphicx} 
\usepackage{algpseudocode}
\usepackage{comment}
\usepackage[english]{babel}
\usepackage{todonotes}
\usepackage{csquotes}
\usepackage{mathtools}
\usepackage{enumitem}

\usepackage{thmtools,thm-restate}

\usepackage{multibib}
\newcites{supp}{Supplementary References}

\usepackage[
 	colorlinks=true,
	urlcolor=black,
	linkcolor=blue
]{hyperref}

\usepackage{amsmath}
\usepackage{amssymb}

\usepackage{algorithm}
\usepackage{xcolor}

\newcommand{\edgeL}{\mathfrak{E}}
\newcommand{\vertexL}{\mathfrak{V}}

\newcommand{\NoEdge}{{\not{\mathtt{E}}}}

\title{On Finding All Connected Maximum-Sized Common Subgraphs in Multiple Labeled Graphs}
\titlerunning{Maximum Common Subgraphs}

\author{Johannes B.S. Petersen\inst{2,1} \and
Akbar Davoodi\inst{1}\thanks{corresponding author: akb@sdu.dk} \and
Thomas Gärtner\inst{2} \and
Marc Hellmuth\inst{3} \and 
Daniel Merkle\inst{4,1} 
}

\authorrunning{Petersen et al.}

\institute{Department of Mathematics and Computer Science,
    University of Southern Denmark, Odense, Denmark
    \and
Faculty of Informatics, Technical University of Vienna, 
\and
Department of Mathematics, Faculty of Science, Stockholm University, Stockholm, Sweden 
\and
Algorithmic Cheminformatics Group, Faculty of Technology, Bielefeld University, Germany 
}

\begin{document}

\maketitle
\begin{abstract}
We present an exact algorithm for computing all common subgraphs with the maximum number of vertices across multiple graphs. Our approach is further extended to handle the connected Maximum Common Subgraph (MCS), identifying the largest common subgraph in terms of either vertices or edges across multiple graphs,
where edges or vertices may additionally be labeled to account for possible atom types or bond types, a classical labeling used in molecular graphs.
Our approach leverages modular product graphs and a modified Bron–Kerbosch algorithm to enumerate maximal cliques, ensuring all intermediate solutions are retained. A pruning heuristic efficiently reduces the modular product size, improving computational feasibility. Additionally, we introduce a graph ordering strategy based on graph-kernel similarity measures to optimize the search process. Our method is particularly relevant for bioinformatics and cheminformatics, where identifying conserved structural motifs in molecular graphs is crucial. Empirical results on molecular datasets demonstrate that our approach is scalable and fast.
\keywords{cheminformatics, subgraph isomorphism; graph similarity and graph kernels; modular product; 
exact algorithms}
\end{abstract}

\sloppy
\section{Introduction}
The maximum common subgraph (MCS) problem stands as a central challenge in numerous domains, including bioinformatics , cheminformatics, pharmacophore mapping, pattern recognition, computer vision, code analysis, compiler design, and model checking \cite{Soul2021,Ehrlich2011,Raymond2002,Conte2011}, to just name a few. Its importance is underscored by applications such as structural alignment in systems biology \cite{Duesbury2018}, where large interaction networks or metabolic pathways must be compared to reveal conserved motifs, or in cheminformatics, where finding a maximal common subgraph helps uncover shared molecular frameworks and acts as a core ingredient for atom-atom mapping or property finding in algorithmic chemistry.

\renewcommand{\thefootnote}{\fnsymbol{footnote}}

Although most foundational work focused on the pairwise MCS problem—which is already NP-complete \cite{Garey1979}—the need to compare multiple labelled graphs naturally arises in contexts like drug discovery and protein analysis. The  MCS problem for more than two graph is significantly more challenging, prompting for heuristics that risk overlooking the truly maximal common subgraph. To address this gap, we propose an approach that can handle reasonably large sets of graphs and provably enumerate all optimal solutions\footnote{For completeness, we include a proof for the unlabeled and not necessarily connected MCS case in the Appendix which is not included in the conference version but available on arXiv 
\cite{arXivversion}. A more technical proof for the labeled and connected case will appear in an extended version of this paper.}. 

Our proposed method is related in spirit to the MultiMCS algorithm of Hariharan et al.\ \cite{MultiMCS}, which also applies modular products for multiple MCS \footnote{A more detailed discussion of related work is given in Appendix D in \cite{arXivversion}}. Hariharan et al.\ proposed the following open challenge: \emph{How do we speed up the computation without losing provable correctness for the vast majority of molecules?} We solve this problem by providing an exact and fast algorithm to list \emph{all}
(connected) maximum-sized common subgraphs in multiple labeled graphs. 
As the code and test data of Hariharan et al.\
are not publicly accessible and no formal proof of correctness is provided, it is not straightforward to compare results. Moreover, the lack of rigorous definitions permits varied interpretations, underlining the complexity of establishing correctness in multiple-MCS settings, an issue our approach seeks to address. 

Our method follows an interative pairwise graph comparison approach but preserves all maximal intermediate solutions to ensure that no global optimum is lost. The underlying pairwise MCS computation employs a modular product construction \cite{Barrow1976} combined with a Bron--Kerbosch-based \cite{Bron-Kerbosch} enumeration. We further introduce a pruning step that removes edges in the modular product without affecting clique connectivity, substantially reducing the search space. 
A key feature of our approach is the investigation of graph ordering strategies. While retaining all intermediate solutions ensures exactness regardless of processing order, selecting the next pair of graphs to be compared can greatly improve efficiency. By leveraging various similarity measures, we demonstrate how to pre-compute an order for the graph products, that minimizes computational overhead at each step.

The remainder of this paper is organized as follows. In Section~\ref{sec:find}, we provide formal definitions required throughout the paper. In Section~\ref{sec:methods}, we describe how pairwise MCS can be computed using 
a modified Bron–Kerbosch-based approach, including the key techniques to improve runtime. 
We evaluate our methods on different data sets in Section~\ref{sec:results}.

\section{Finding maximum common subgraphs\label{sec:find}}

We consider simple undirected graphs $G = (V, E)$ with a finite vertex set $V(G) \coloneqq V$ and an edge set $E(G) \coloneqq E$. We write $G \simeq H$ if the graphs $G$ and $H$ are isomorphic. 
A \emph{``pure'' subgraph} of a graph $G$ is a graph $H$ satisfying $V(H) \subseteq V(G)$ and $E(H) \subseteq E(G)$. Moreover, $H$ is a \emph{subgraph} of $G$ if $G$ contains a pure subgraph $H'$ such that
$H\simeq H'$. 
If $H$ is a subgraph of both $G_1$ and $G_2$, then it is called a \emph{common subgraph}. A common subgraph is a \emph{maximum-vertex common subgraph (MVCS)} of $G_1$ and $G_2$ if $|V(H)|$ is maximized among all common subgraphs of $G_1$ and $G_2$. Similarly, a \emph{maximum-edge common subgraph (MECS)} $H$ is a common subgraph that maximizes $|E(H)|$. The latter definition naturally generalizes to common subgraphs of more than two graphs. 

A classical method to find an MVCS of two graphs $G_1$ and $G_2$ is as follows \cite{Koch}. First, compute the modular product $G_1 \star G_2$ with vertex set $V(G_1 \star G_2) \coloneqq V(G_1) \times V(G_2)$, where the edge set $E(G_1 \star G_2)$ includes all pairs $\{(u,v),(u',v')\}$ for which either $\{u,u'\}\in E(G_1)$ and $\{v,v'\}\in E(G_2)$ or
$\{u,u'\}\notin E(G_1)$ and $\{v,v'\}\notin E(G_2)$.
Each vertex $(u, v)$ in a clique corresponds to vertex $u$ in $G_1$ and vertex $v$ in $G_2$, and an edge $\{(u, v), (u', v')\}$ in the clique implies that either both $\{u, u'\}$ and $\{v, v'\}$ are edges or both are non-edges in $G_1$ and $G_2$, respectively. Therefore, we can find all maximal common induced subgraphs between two graphs by finding all maximal cliques in their modular product \cite{Levi1973ANO}.

The modular product can naturally be generalized to more than two factors. However, the product is not associative in general, i.e. $({G_1} \star {G_2}) \star {G_3} \not\simeq {G_1} \star ({G_2} \star {G_3})$ may happen. To simplify notation, we use the convention that $\star_{i=1}^n G_i \coloneqq (((G_1 \star G_2) \star G_3) \dots \star G_n)$ and that $v = (v_1, v_2, v_3, \dots, v_i, \dots, v_n) \coloneqq (((v_1, v_2), v_3) \dots, v_i), \dots, v_n) \in V(\star_{i=1}^n G_i)$. Given $G = \star_{i=1}^n G_i$, the \emph{projection onto the $i$-th factor} is the map $p_i \colon V(G) \to V(G_i)$, defined by $p_i((v_1, v_2, v_3, \dots, v_i, \dots, v_n)) = v_i$. For a given subgraph $H \subseteq G$, we often write $p_i(H)$ for the subgraph of $G_i$ induced by the vertices $p_i(v)$ with $v \in V(H)$.

As outlined above, an MVCS of $G_1$ and $G_2$ can be found by filtering all inclusion-maximal cliques in $G_1 \star G_2$. We aim to extend this idea to establish methods for:

\begin{itemize} 
\item[(1)] Finding an MVCS for more than two graphs. 
\item[(2)] Finding a \emph{connected} MVCS for more than two graphs. 
\item[(3)] Finding a (connected) MECS for two or more graphs. 
\item[(4)] Handling the MVCS and MECS problems for edge-labeled and vertex-labeled graphs, where common subgraphs are determined by retaining the respective labels. \end{itemize}

For Task (1), assume that we want to compute an MVCS of the graphs $G_1, G_2, \dots, G_n$. We first compute all inclusion-maximal cliques in $G_1 \star G_2$ and store them in the set $\mathcal{K}$. Then, for each $K \in \mathcal{K}$, we consider the subgraphs $H^i_K = p_i(K)$ in $G_i$, $1 \leq i \leq 2$. By definition of the modular product, we have $H^1_K \simeq H^2_K$ for each $K \in \mathcal{K}$. We then collect a representative for each $K \in \mathcal{K}$, defined as $H_K \coloneqq H^2_K$, and store all representatives in the set $\mathcal{H}_{1,2}$. For each $H \in \mathcal{H}_{1,2}$, we repeat this process: we determine the set of all inclusion-maximal cliques in $H \star G_3$ and obtain the new set $\mathcal{K}$ as the collection of all inclusion-maximal cliques in $H \star G_3$ for all $H \in \mathcal{H}_{1,2}$. Taking the representatives $p_3(K)$ of all such cliques $K \in \mathcal{K}$ yields the set $\mathcal{H}_{1,2,3}$. We repeat this process recursively until we derive the set $\mathcal{H}_{1, \dots, n}$. The next two results 
(proofs are provide in \cite{arXivversion}) show that this approach provides exact solutions for the MVCS problem and can be used to determine all MVCS of $G_1, \dots, G_n$.

\begin{restatable}{lemma}{xlemma}
\label{lem:common-1}
	Let $G_1,\dots, G_n$ be graphs. 
	If $H$ is a maximal common  induced subgraph of $G_1,\dots,G_n$, then there exists a
	maximal clique $K$ in $\star_{i=1}^n G_i$ of size $|V(K)| = |V(H)|$
    and for which the projection $p_i$ onto the $i$-th factor satisfies
    $p_i(K) = H_i\simeq H$.
\end{restatable}
\vspace{-0.1in}
\begin{restatable}{proposition}{primeProp}
     \label{prop:main}
 	Let $G_1,\dots, G_n$ be graphs and let $\mathcal{K}$ be the set of all maximal cliques $K$ in $\star_{i=1}^n G_i$. 
 	In addition, let $H$ be the subgraph induced by the vertex set $\{p_i(w)\mid w\in V(K)\}$ in $G_i$  
 	for some fixed	 $K\in \mathcal{K}$.  Then the following two statements are equivalent.
	\begin{enumerate}
	\item $H$ is a  MVCS subgraph of $G_1,\dots,G_n$.
	\item $K$ is a maximum clique within the set $\mathcal{K}$.
	\end{enumerate}
\end{restatable}

By Lemma~\ref{lem:common-1}, every maximal common induced subgraph of $G_1, \dots, G_n$ corresponds to some maximal clique $K$ in $\star_{i=1}^n G_i$ of size $|V(K)| = |V(H)|$, and it holds that $p_i(K) = H_i \simeq H$. Since \emph{all} maximal cliques are collected in each step and stored in the set $\mathcal{K}$, the set $\mathcal{K}$ contains also \emph{all} maximum cliques in $\star_{i=1}^n G_i$. This, together with Proposition~\ref{prop:main}, implies that the set $\mathcal{H}_{1, \dots, n}$ contains all (up to isomorphism) MVCS of $G_1, \dots, G_n$.

Now consider Task (2), where we want to find a \emph{connected} MVCS of the graphs $G_1, G_2, \dots, G_n$. 
To recall, edges $\{(u, v), (u', v')\}$ in $G_1 \star G_2$ are formed by either edges $\{u, u'\} \in E(G_1)$ and $\{v, v'\} \in E(G_2)$, or non-edges $\{u, u'\} \notin E(G_1)$ and $\{v, v'\} \notin E(G_2)$. Thus, we can additionally label the edges $\{(u, v), (u', v')\}$ as \texttt{type-1} if they are based on edges in the factors, or as \texttt{type-0} otherwise.
Assume now that we have found a clique $K$ in $G_1 \star G_2$  that may consist of both \texttt{type-0} and \texttt{type-1} edges. 
In \cite{Koch}, such edges are also called $c$-edges (connected edges) and d-edges (disconnected edges).
If there exists a $uv$-path in $K \subseteq G_1 \star G_2$ that consists only of \texttt{type-1} edges, it is easy to verify that there is a $p_i(u)p_i(v)$-path in $G_i$ for $i \in {1, 2}$. We call a clique $K$ \emph{\texttt{type-1} connected} if, for all $u, v \in V(K)$, there exists such a \texttt{type-1} edge $uv$-path. For \texttt{type-1} connected cliques $K \subseteq G_1 \star G_2$, it holds that $p_i(K)$ is a connected subgraph of $G_i$ for $i \in \{1, 2\}$. This generalizes naturally to more than two factors. Hence, if we restrict our attention to \texttt{type-1} connected cliques, we can reuse the approach from Task (1) to obtain all \emph{connected} MVCS of $G_1, \dots, G_n$.

For the MECS problem in Task (3), consider the line graphs $L(G) = (W, F)$ of a given graph $G = (V, E)$, where $W \coloneqq E$ and $\{e, f\} \in F$ if and only if $e \neq f$ and $e \cap f \neq \emptyset$. It is well-known \cite{Koch} that finding a (connected) MECS in $G$ is equivalent to finding a (connected) MVCS in $L(G)$, with special handling required for claw and triangle graphs. Using the methods from Tasks (1) and (2), we can start with $L(G_1)$ and $L(G_2)$ and compute $L(G_1) \star L(G_2)$ to find all maximal (connected) MVCS. We then continue in a natural manner with $L(G_3), \dots, L(G_n)$ to find all maximal (connected) MVCS of $L(G_1), \dots, L(G_n)$, which translates to all maximal (connected) MECS of $G_1, \dots, G_n$.

Finally, we consider Task (4), where we may also track specific atom types (vertex labels) or bond types (edge labels). In this case, $H$ is a labeled MVCS or MECS if it is an induced subgraph of all $G_1, \dots, G_n$ while preserving all edge and vertex labels.  
To achieve this, we extend the \texttt{type-1} classifications in the product by introducing additional types based on edge and vertex labels while keeping 
\texttt{type-0} edges unchanged. Specifically, we define an edge or vertex as \texttt{type-A} if it  is  \texttt{type-1} and 
has, in addition, a label from a predefined set $\texttt{A}$ of allowed types. For example, edge labels may distinguish between double and single hydrogen bonds, while vertex labels may represent different atom types.  
For non-labeled graphs, the set $\texttt{A}$ defaults to $\{\texttt{1}\}$, defining \texttt{type-1} and \texttt{type-0} edges, while vertices remain unlabeled. In labeled graphs, connectedness refers to paths containing only edges or vertices of a particular type. To be more precise, a graph is \texttt{type-A}-connected if a path exists between any two vertices consisting solely of \texttt{type-A} edges. Note, \texttt{type-A}-connected clique are also knows \textit{c-clique}, see \cite{Koch}.  
 If $G_1$ and $G_2$ are $\texttt{A}$-labeled graphs, their labels transfer to the product $G_1 \star G_2$ by ensuring that an edge $e = \{(u,v),(u',v')\}$ or a vertex $x = (u,v)$ in $G_1 \star G_2$ inherits label $\ell \in \texttt{A}$ if both corresponding edges $\{u,u'\}$ and $\{v,v'\}$, or vertices $u$ and $v$, in the factors share label $\ell$. Otherwise, $e$ or $x$ are marked, for instance, as inconsistent.
Clique finding in $\texttt{A}$-labeled products follows the same approach as in unlabeled products, with a simple relabeling step: 
``consistent'' \texttt{type-A}  edges in the product that correspond to edges in the factors are reassigned as \texttt{type-1}, 
edges in the product that correspond to non-edges in the factors are reassigned as \texttt{type-0}
while all ``inconsistent'' edges are simply removed  and do not appear as edges in the product.

\section{Methods}\label{sec:methods}

Here, we present the algorithmic choices implemented to ensure correctness and efficiency in solving the problem. All methods have been implemented in \texttt{C++} using the boost graph and  openbabel library.
The source code is freely available at \url{https://github.com/johannesborg/cmces}.

\vspace{-0.1in}\paragraph{\textbf{Pruning steps.}}

The methods established here implicitly create a search tree, where all leaves are potential maximum
common subgraphs. It turns out that a lot of branches of this search tree can be pruned by
performing a depth-first search and finding a candidate maximum common subgraph and then
backtracking and pruning the search tree based on the current best candidate.
The basic idea is as follows. To recall,  $\mathcal{H}_{1,\dots, j}$ denotes
the candidate sets of the MVCS of the first $j$ graph $G_1,\dots,G_j$. 
Here, we compute first $\mathcal{H}_{1,2}$ as outlined in Section~\ref{sec:find}. 
Afterwards, we take one $H\in  \mathcal{H}_{1,2}$ compute a candidate set
$\mathcal{H}'_{1,2}$ only for $H\star G_3$ which results in a subset 
$\mathcal{H}'_{1,2,3}\subseteq \mathcal{H}_{1,2,3}$. In this way, we obtain
a subset $\mathcal{H}'_{1,\dots,n}\subseteq \mathcal{H}_{1,\dots,n}$.
Now, we can remove all elements $H\in \mathcal{H}_{1,2}$ for which 
$|V(H)|< m$ where $m$ denotes the size of elements  $H'\in \mathcal{H}'_{1,\dots,n}$ for which $|V(H')|$ is maximum. 
This, in turn, can heavily reduce the size of elements in $\mathcal{H}_{1,2}$
and, thus, the number of candidates on which we need to recurse.

\vspace{-0.1in}\paragraph{\textbf{Techniques for finding  \texttt{type-A} connected cliques.}}
The Bron–Kerbosch algorithm \cite{Bron-Kerbosch} is a clique-finding algorithm that iteratively extends non-maximal cliques until they can be certified as maximal, at which point they are returned. Based on variants of the Bron–Kerbosch algorithm, Koch’s algorithm (Alg.~3 in \cite{Koch}) can be used to find all connected maximal subgraphs of two graphs. We extend both algorithms to find \texttt{type-A} connected cliques, and consequently, \texttt{type-A} connected MVCS and MECS in multiple graphs.  

In the \texttt{type-A} connected MVCS and MECS problems, we seek not just maximal cliques in $\star_{i=1}^n G_i$ and $\star_{i=1}^n L(G_i)$, but specifically maximal \texttt{type-A} connected cliques. We illustrate the key ideas using the MVCS problem on $\star_{i=1}^n G_i$, as the same approach naturally extends to $\star_{i=1}^n L(G_i)$ for the MECS problem.  

At each step $j$, we take candidates $H \in \mathcal{H}_{1\dots,j-1}$, compute $G = H \star G_j$, and identify all \texttt{type-A} connected cliques in $G$. Since such cliques are contained within \texttt{type-A} connected subgraphs of $G$, we first partition $G$ into its \texttt{type-A} connected components $G_1', \dots, G_k'$ by retaining only \texttt{type-A} edges. We then augment each $G_i'$ by restoring all non-\texttt{type-A} edges, yielding induced subgraphs $G_1'', \dots, G_k''$ that remain \texttt{type-A} connected. To find all \texttt{type-A} connected maximal common subgraphs of $H$ and $G_j$, it suffices to compute the maximal \texttt{type-A} connected cliques in each $H \star G_i''$ separately, enabling parallel processing.

As outlined at the end of Section~\ref{sec:find}, for finding cliques in $\texttt{A}$-labeled products, we relabeled  ``consistent'' \texttt{type-A} edges as  \texttt{type-1} edges, 
all edges in the product referring to non-edges in the factors as  \texttt{type-0}
while all other edges in the product have been removed. 
To further optimize runtime, we remove additional \texttt{type-0} edges in $G = H \star G_j$ before determining the \texttt{type-A} connected components $G_1'', \dots, G_k''$. Specifically, we remove \texttt{type-0}  edges  $\{(u,v), (u',v')\}$ from $G = H \star G_j$ whenever there exists no \emph{simple} \texttt{type-A} $uv$-path in $H$ that is isomorphic to a simple \texttt{type-A} $u'v'$-path in $G_j$. Those particular \texttt{type-0} edges  $\{(u,v), (u',v')\}$ indicate that any \texttt{type-A} connected  subgraph of $H$ containing $u$ and $u'$ cannot be isomorphic to any \texttt{type-A} connected subgraph of $G_j$ containing $v$ and $v'$, meaning $(u,v)$ and $(u',v')$ cannot be part of the same \texttt{type-A} connected clique corresponding to a maximal \texttt{type-A} connected common subgraph of $H$ and $G_j$. 
When partitioning $G$ into \texttt{type-A} connected components $G_1', \dots, G_k'$ (as outlined in the preceding part), we include only ``consistent'' \texttt{type-A} edges. 
Removal of the additional \texttt{type-0} edges as explained here yields a finer
partition $G_1'', \dots, G''_{k'}$. 
We then compute maximal \texttt{type-A} connected cliques in $H \star G_i''$. Although these \texttt{type-0} edge-removal operations are based on computationally hard problems, Section~\ref{sec:results} shows that they significantly speed up clique detection—likely due to the moderate size of the molecular graphs under consideration.

\vspace{-0.1in}\paragraph{\textbf{Handling triangles and claws.}}

A crucial step in our method is handling special cases when finding the MECS in $G = \star_{i=1}^n G_i$ by identifying maximal cliques in $L(G)$. This issue, known as the $\Delta$-Y exchange problem \cite{Topsim,Pascal}, arises due to structural ambiguities in line graphs.  

Let $K_3$ denote the complete graph on three vertices (a triangle) and $K_{1,3}$ the claw graph, consisting of four vertices and three edges meeting at a single vertex. It is well-known that $L(K_3) = K_3$ and $L(K_{1,3}) = K_3$, making them the only two graphs indistinguishable by their line graphs \cite{whitney}. Consequently, if input graphs contain both triangles and claws before conversion to line graphs, naive methods may incorrectly report non-existent common subgraphs.  
To resolve this, we use the inverse mapping $\delta_i$, which maps vertices in $V(L(G_i))$ back to their corresponding edges in $E(G_i)$. When a $K_3$ is found in $\star_{i=1}^n L(G_i)$, we define $H_i$ as the subgraph of $G_i$  whose edges refer to the vertices of this $K_3$ in $L(G_i)$. If $H_i \not\simeq H_j$ for some $G_i$ and $G_j$, one the graphs
$G_i$ and $G_j$ contains a claw while the other contains a triangle. We classify such $K_3$ structures in $\star_{i=1}^n L(G_i)$ as ``bad'' triangles.  
Instead of including a bad triangle $T \simeq K_3$ in $\star_{i=1}^n L(G_i)$, we replace it with its three subgraphs, $T - v_1$, $T - v_2$, and $T - v_3$. These derived subgraphs serve as refined candidates for the MECS of $G_1, \dots, G_n$, ensuring that only valid common subgraphs are considered.

\vspace{-0.1in}\paragraph{\textbf{Ordering of input graphs.}}

As we will see in Section~\ref{sec:results}, the ordering of input graphs significantly impacts the algorithm's runtime. Consider three graphs ordered as $[G_1, G_2, G_3]$ versus $[G_1, G_3, G_2]$. If the number of maximal common subgraphs between $G_1$ and $G_2$ is much larger than between $G_1$ and $G_3$, then $|\mathcal{H}_{12}| > |\mathcal{H}_{13}|$. Since we compute $H \star G_3$ for each $H \in \mathcal{H}_{12}$ and $H \star G_2$ for each $H \in \mathcal{H}_{13}$, the ordering $[G_1, G_3, G_2]$ results in fewer graphs requiring clique detection, leading to improved efficiency.  
However, determining an optimal ordering a priori is challenging, as determining maximal connected common subgraphs between each pair of graphs is itself NP-hard \cite{Garey1979}. To address this, we use an efficient heuristic that greedily selects graph pairs with low similarity, based on the intuition that graphs with lower similarity tend to have fewer maximal common subgraphs. Our procedure for ordering input graphs works as follows. Given a set of input graphs, $\{G_1,\dots,G_n\}$, and a similarity measure $\kappa \colon \mathcal{G} \times \mathcal{G} \rightarrow \mathbb{R}$, we first select the two graphs, $G_i$ and $G_j$, which minimize $\kappa_{ij} \coloneqq \kappa(G_i,G_j)$ over all pairs of graphs. We let the tentative ordering be $O=[G_{i}, G_{j}]$. Then we proceed as follows: given a tentative ordering $O = [G_{l_1},\dots,G_{l_k}]$ we select the graph $G_p \in \{G_1,\dots,G_n\}-O$, which minimizes $\max\{\kappa_{l_1p},\dots,\kappa_{l_kp}\}.$ We update the tentative ordering to include $G_p$, i.e. $O=[G_{l_1},\dots,G_{l_k}, G_p]$. We obtain our final ordering $O$ when $\{G_1,\dots,G_n\}-O = \emptyset$. 

For similarity estimation, we employ two pairwise graph similarity measures.
One of them is based on graph kernels \cite{DBLP:journals/corr/abs-1903-11835}, functions $\kappa \colon \mathcal{G} \times \mathcal{G} \rightarrow \mathbb{R}$ widely used in machine learning. While standard graph kernels must be positive semi-definite, we do not impose this requirement. In particular, we used three of the standard graph kernels: 
the vertex histogram kernel(VH) \cite{Nikolentzos_2021}, the Weisfeiler-Lehman optimal assignment kernel (WL) \cite{DBLP:journals/corr/abs-1903-11835} and the neighborhood subgraph pairwise distance kernel (NSPD) \cite{Nikolentzos_2021}).
The other similarity measure, called the ``minmax'' similarity, is defined as follows. 
We compute the minmax similarity of two graphs, $G_i$ and $G_j$, by first partitioning $G_i \star G_j$ into its \texttt{type-A} connected components, $G_1', \dots, G_k'$ as outlined in paragraph \emph{Techniques for finding \texttt{type-A} connected cliques}. 
The minmax similarity of $G_i$ and $G_j$ is then computed as $\kappa_{ij} \coloneqq \max\big\{|V(G_\ell')|  \mid \ell = 1,\dots,k\big\}$. Using the minmax similarity measure with the greedy ordering as outlined above, we greedily minimize the size of the maximum sized graph 
in which we have to find cliques in each of the steps. 
The computational results in Section~\ref{sec:results} demonstrate a strong relationship between the chosen similarity measures and the number of maximal connected common subgraphs, validating the choice of these measures.

\section{Results\label{sec:results}} 
To understand the relationship between similarity values between graphs and the number of their maximal \texttt{type-A} connected common edge subgraphs, we conducted the following experiment. For 500 molecular graphs, $\mathcal{M} = \{G_1,\dots,G_{500}\}$, from the ZINC data set \cite{doi:10.1021/acs.jcim.5b00559}, each having $n$ vertices with $20\leq n\leq 25$, we computed for each pair of graphs their number of maximal  \texttt{type-A} connected common edge subgraphs, and also the similarities measures (the three graph kernels VH, WL, NSDP and the minmax similarity - see Section~\ref{sec:methods}).

\begin{figure}[htbp]
    \centering
    \includegraphics[width=1.\textwidth]{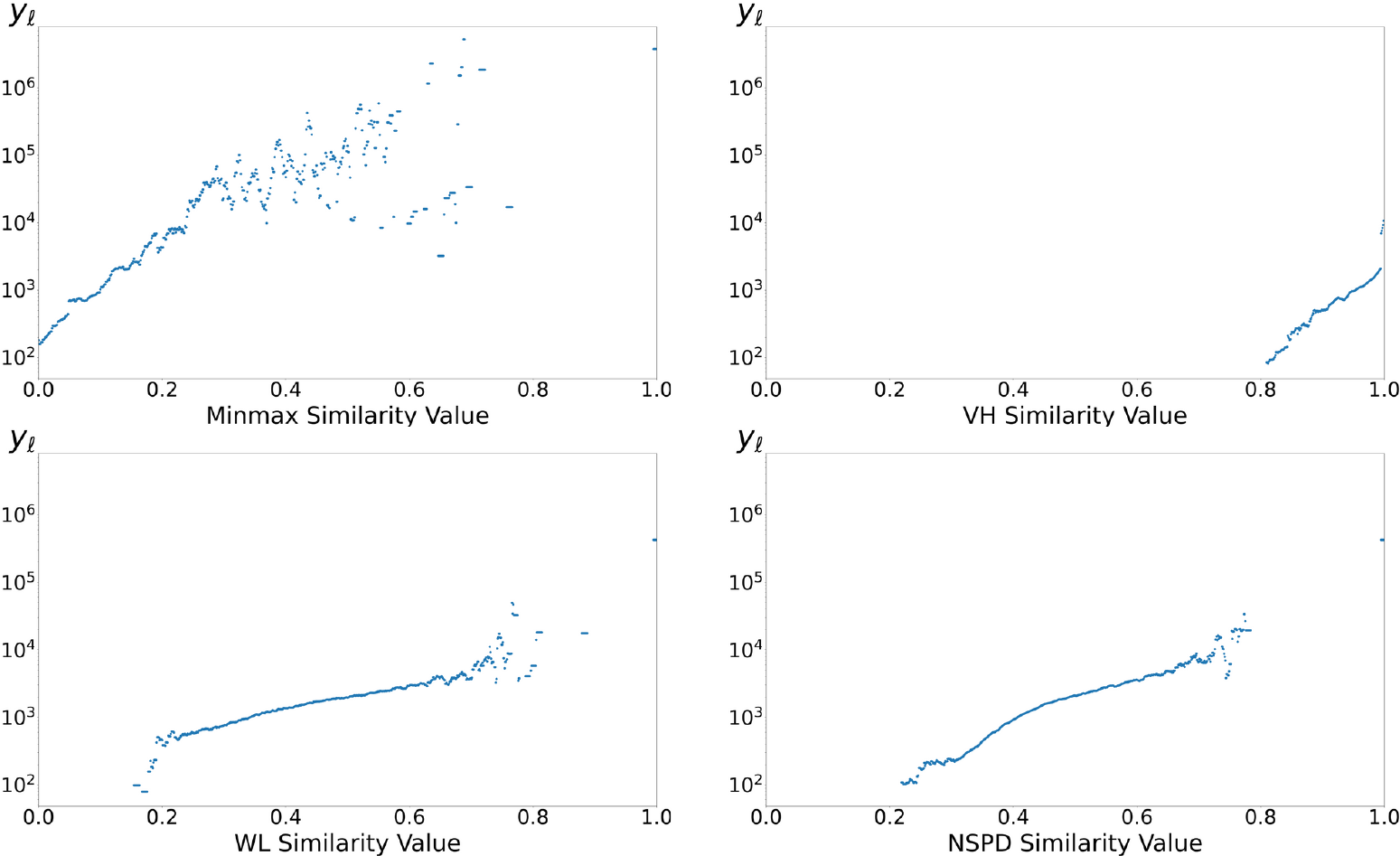}
    \caption{
    Plots showing the relationship between the four normalized similarity measures VH, WL, NSDP or          minmax (x-axis) and the values $y_\ell$ on the y-axis that reflect average number of           \texttt{type-A} connected  common subgraphs of two graphs, see text for further details. }
    \label{fig:kernel_analysis}
\end{figure}

The relationship between \texttt{type-A} connected common subgraphs of $G_i$ and $G_j$ and the respective similarity between
$G_i$ and $G_j$ is shown in Figure \ref{fig:kernel_analysis}. To avoid getting the plot scattered by outliers
we decided to streamline the plot by adding pairs of graphs from $\mathcal{M} $ into buckets. 
To be more precise, let $\kappa_{ij}$ be one of the four normalized similarity measures VH, WL, NSDP or minmax between $G_i, G_j\in \mathcal{M}$. Moreover, let $m_{ij}$ be the number of inclusion-maximal cliques in the product  $G_i\star G_j$ and, therefore, the number of \texttt{type-A} connected common subgraphs of $G_i$ and $G_j$ induced by inclusion-maximal \texttt{type-A} connected cliques in $G_i\star G_j$. 
We now define a set of 1,000 buckets, $B = \{b_1,\dots,b_{1000}\}$
where each bucket $b_\ell \coloneqq \{m_{ij} \mid \kappa_{ij} \in [ \ell \cdot0.001 + 0.005, \ell \cdot0.001-0.005] \}$
contains the the value $m_{ij}$ whenever the respective similarity $\kappa_{ij}$ between $G_i$ and $G_j$
is contained in the prescribed interval. In other words, the buckets serve as a sliding window
along the x-axis in which we collect the pairs of graphs $G_i$ and $G_j$ based on their
values $m_{ij}$ and $\kappa_{i,j}$. For each nonempty bucket $b_{\ell}$ we compute the pair $(x_\ell,y_\ell)= (\ell\cdot0.001, \frac{\sum_{m\in b_\ell}m}{|b_\ell|})$. 
The value $x_\ell = \ell\cdot0.001$  is a representative of the similarity values
$\kappa_{ij} \in [ \ell \cdot0.001 + 0.005, \ell \cdot0.001-0.005]$
while $y_\ell$ is the average of the number of inclusion-maximal \texttt{type-A} connected cliques of in the product of pairs of
graphs $G_i$ and $G_j$ whose value $m_{ij}$ is contained the bucket $b_\ell$.
The plot showing the pairs $(x_\ell,y_\ell)$ computed over all pairs of graphs in $\mathcal{M}$
is provided in Figure~\ref{fig:kernel_analysis}. 
Roughly spoken, Figure \ref{fig:kernel_analysis} shows the average number of inclusion-maximal \texttt{type-A} connected cliques in the product  $G\star H$ (and, therefore, the average number of \texttt{type-A} connected common subgraphs of $G$ and $H$ induced by  those cliques in $G\star H$) 
and the respective similarity value between $G$ and $H$ for all pairs $G,H\in \mathcal{M}$.

\begin{figure}[htbp]
    \begin{center}
    \includegraphics[width=1.\textwidth]{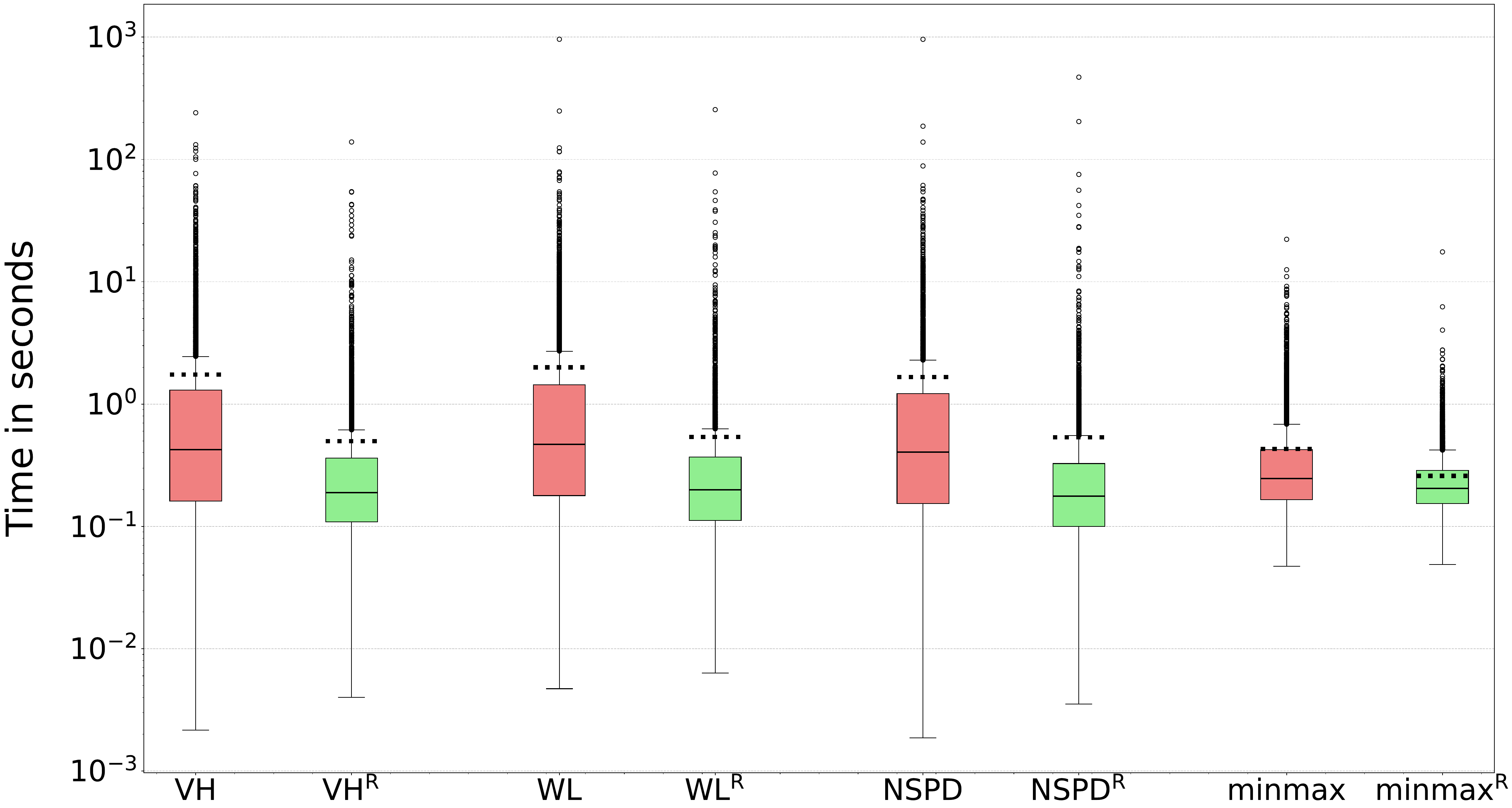}
    \end{center}\vspace{-0.1in}
    \caption{Box-plot showing the effect of different techniques applied to the input graphs
    (instances of 5 molecular graphs with  35 atoms) on the runtime for finding MECS. 
    In particular, we applied the greedy-ordering based on the four similarity 
    measure VH, WL, NSPD and minmax together with the computation of maximal \texttt{type-A} connected cliques with and without the removal of \texttt{type-0} edges (see Section~\ref{sec:methods}). The term $Z\in \{\text{VH, WL, NSPD, minmax}\}$
    on the $x$-axis refers to the application of measure $Z$ without this refinement step, while
    $Z^R$ means that the removal of certain \texttt{type-0} edges has been applied. The dotted lines represent  the mean values across all instances. Without greedy-ordering, 
    the runtime exceeded in many cases one hour, if it finished at all.}
    \label{fig:runtime_analysis}
\end{figure}

Figure~\ref{fig:kernel_analysis} illustrates a clear relationship between the VH kernel, the WL kernel, the NSPD kernel, and the minmax similarity with the number of common subgraphs. Specifically, higher similarity values correspond to a greater number of common subgraphs induced by inclusion-maximal cliques in the respective products. These observations highlight that these similarity measures are indeed beneficial for the methods proposed in the paragraph \emph{Ordering of Input Graphs}, which rely on such measures.

\begin{figure}[htbp]
    \label{mols}
    \centering
    \includegraphics[width=1\textwidth]{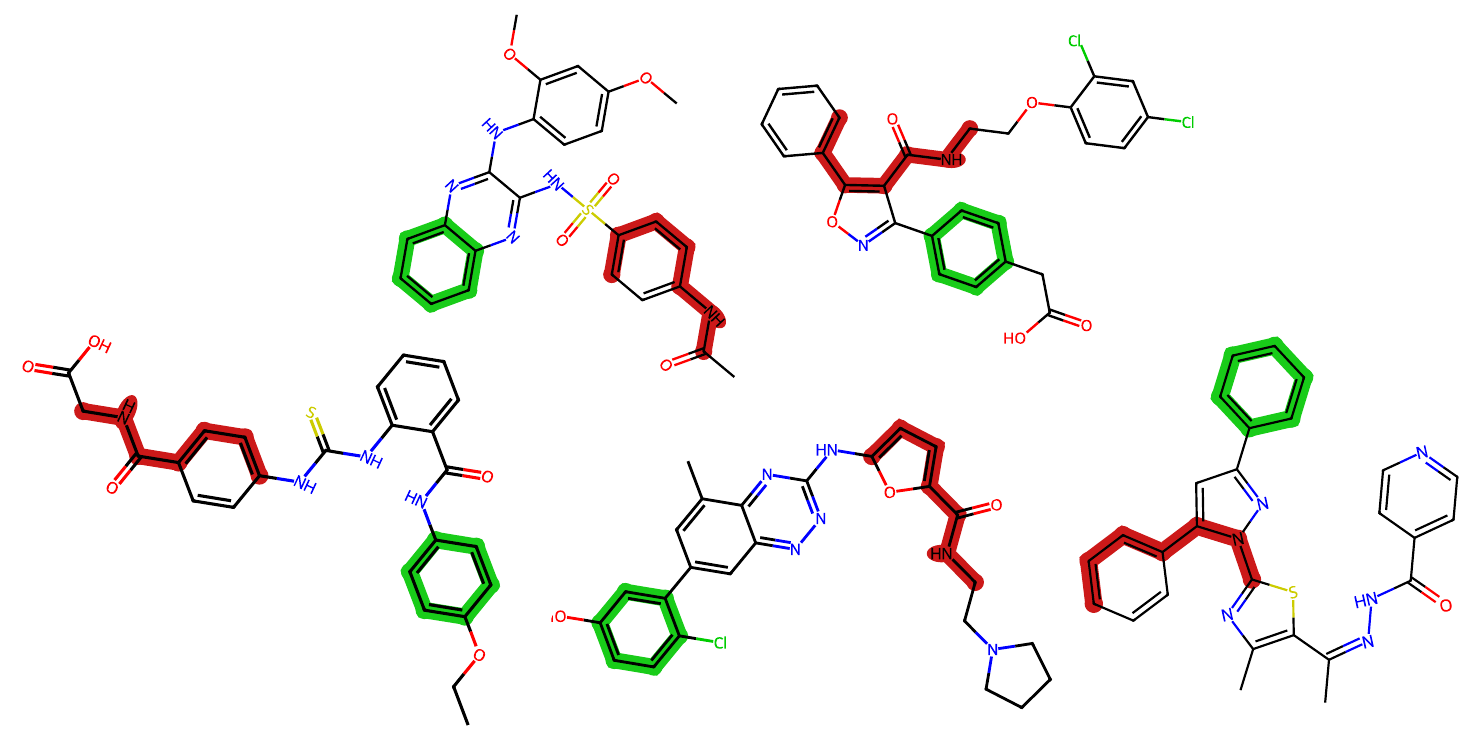}
    \caption{An example instance with five different molecular graphs from the ChEMBL22 database with 35 non-hydrogen vertices each. There are exactly two MECS with 6 edges each. A single occurrence of the
    two MECS within each graph is highlighted in red and green, respectively. Note that more occurrences 
    are possible.}
    \label{fig:mols}
\end{figure}

To test the effect of ordering the input graphs together with 
``removal of certain \texttt{type-0} edges'' for finding MECS as outlined in the last part of 
paragraph \emph{Techniques for finding  \texttt{type-A} connected cliques}, we created 5000 problem instances, each containing 5 molecular graphs with  35 atoms (excluding hydrogen)
randomly chosen from the ChEMBL22 database \cite{10.1093/nar/gkad1004}.  All experiments were performed on a machine with 32 GB of LPDDR5X RAM, an Intel Core Ultra 9 processor and 1TB NVMe SSD, running the Fedora Workstation 41 operating system. 
To determine an ordering of the input graphs we used the four similarity measures VH, WL, NSPD and minmax
and the method outlined in the paragraph \emph{Ordering of Input Graphs}. For each ordering, we determined the maximal \texttt{type-A} connected cliques with and without 
the removal of certain \texttt{type-0} edges. Each of the eight different configurations
have been applied in the 5000 test instances and the run-time in seconds has been recorded. 
The resulting box plots are shown in Figure~\ref{fig:runtime_analysis}.
We emphasize that Figure~\ref{fig:runtime_analysis} does not include box plots for randomly ordered graphs because, in many cases, the runtime for computing the maximal \texttt{type-A} connected cliques exceeded one hour ($3.6 \times 10^3$ seconds)—if they finished at all.
Furthermore, we highlight that the ordering used in the approach established in \cite{MultiMCS} is based on the size of the molecular graphs. However, since all test instances contain 35 vertices, this ordering is expected to be essentially random. As a result, its runtime can also be assumed to exceeded one hour—if it finished at all.
In contrast, the mean runtime for graphs ordered by similarity ranges from $0.25$ to $1.99$ seconds, highlighting the significant impact of the ordering methods. Moreover, it can be observed that the runtime improves in all cases when applying the removal of the specified \texttt{type-0} edges.
The minmax ordering with the applied refinement step achieves the lowest mean runtime. Furthermore, even without the refinement step, minmax ordering still results in the lowest mean runtime compared to the other ordering methods, regardless of whether refinement is applied.

In Figure \ref{fig:mols} we have shown 5 molecular graphs and highlighted two different labeled MECS 
in red and green, respectively.

\section{Conclusion}\label{sec:conclusion}

We have presented an exact algorithm for finding {\em all} \texttt{type-A} connected maximum common subgraphs, i.e., those maximum subgraphs that preserve connectedness  via certain labeling-constraints,
of any finite collection of labeled graphs. Our approach relies on combining a pairwise MCS subroutine (realized here via the modular product and a suitably adapted Bron--Kerbosch algorithm) together with several techniques to improve runtime, such as graph-ordering based on similarity measures and techniques for finding \texttt{type-A} connected cliques. Indeed, the experiments show that these techniques significantly speed up the runtime of the methods in practice.
We have developed a publicly available implementation of these methods (\url{https://github.com/johannesborg/cmces}), demonstrating its scalability on molecular datasets.

\begin{appendix}

\section*{Appendix}

Here, we bring together the essential background and supporting material that underpins our work, making the paper self-contained and easier to access. In Section~A, we introduce the basic definitions and notations from graph theory that form the foundation of our approach. In Section~B, we use these foundations to  prove Lemma~\ref{lem:common-1} and Proposition~\ref{prop:main}, employing key properties of graph isomorphisms, and cliques in the modular product. Section~C extends our discussion to edge- and vertex-labeled graphs, a framework especially relevant for modeling molecular structures where vertices represent different atom types and edges indicate specific bond types. Finally, Section~D briefly surveys related works, placing our contributions within the broader context of maximum common subgraph research. Throughout the appendix, we follow the notation and conventions as defined in \cite{west2001introduction}, ensuring consistency with established literature.

\section{Basic Definitions}
A \emph{subgraph} of a graph $G = (V, E)$ is a graph $H = (V_H, E_H)$ satisfying $V_H \subseteq V$ and $E_H \subseteq E$, with the additional requirement that every edge in $E_H$ has both endpoints in $V_H$.

An important special case is the \emph{induced subgraph}. For any subset $S \subseteq V$, the induced subgraph of $G$ on $S$, denoted by $G[S]$, is defined by setting $V_{G[S]} = S$ and $E_{G[S]} = \big\{ \{u,v\} \in E \mid u \in S, v \in S \big\}.$
In essence, $G[S]$ contains all vertices in $S$ and every edge of $G$ whose endpoints are both contained in $S$.

Let $G = (V, E)$ and $H = (W, F)$ be two graphs. An \emph{isomorphism} between $G$ and $H$ is a bijection $\varphi: V \to W$ that preserves adjacency; that is, for all $u,v \in V$,
$\{u,v\} \in E$ if and only if $\{\varphi(u), \varphi(v)\} \in F$.
When such a bijection exists, we say that $G$ and $H$ are \emph{isomorphic} and write $G \cong H$.

We also recall the definition of a \emph{clique} in a graph $G$. A clique is a subset of vertices $C \subseteq V$ with the property that for every two distinct vertices $u,v \in C$, the edge $\{u,v\}$ is present in $E$. A clique $C$ is termed \emph{inclusion-maximal} if no vertex $w \in V \setminus C$ can be added to $C$ while preserving the clique property; that is, there exists no $w$ for which $C \cup \{w\}$ is also a clique.

Finally, a \emph{$uv$-path} in $G$ is defined as a sequence of vertices $u = v_0, v_1, \ldots, v_k = v$,
such that for every index $0 \leq i < k$, the edge $\{v_i, v_{i+1}\}$ belongs to $E$.

\section{Proof of Lemma~\ref{lem:common-1}  and Proposition~\ref{prop:main}}
In the following sections, we use the definitions above to prove Lemma~\ref{lem:common-1} and Proposition~\ref{prop:main}.


\xlemma*
\begin{proof}
    Let $G_1,\dots, G_n$, be graphs. Put $G\coloneqq \star_{i=1}^n G_i$. Suppose first that $H$
    be a maximal common induced subgraph of $G_1,\dots,G_n$. By Definition, for
    every $ i$, there is an isomorphism $f_i$ from $H$ to $H_i\subseteq G_i$. 
    Note that $H\simeq H_i$ for all $i\in \{1,\dots,n\}$ implies that
    $H_i\simeq H_j$ for all $i,j\in \{1,\dots,n\}$. Now consider the map $\varphi_{i}\colon
    V(H_1)\to V(H_i)$ defined by $\varphi_i\coloneqq f_i \circ f_{1}^{-1}$ for all $i\in
    \{1,\dots,n\}$. It is easy to verify that $\varphi_{i}$ is a subgraph isomorphism between
    $H_1\subseteq G_1$ and $H_i\subseteq G_i$
    for all $i\in \{1,\dots,n\}$. Note that $\varphi_{1}$ is the identity map on $V(H_1)$, i.e.,
    $\varphi_{1}(w)=w$ for all $w\in V(H_1)$. Consider now the subgraph $K$ of $G$ that is induced
    by $\{(u, \varphi_{2}(u), \dots, \varphi_{n}(u)) \mid u\in V(H_1) \}$.  In particular,
    $|V(K)| = |V(H_1)| = |V(H)|$. Note that the projection $p_1$ satisfies 
    $p_i(u, \varphi_{2}(u), \dots, \varphi_{n}(u))$ for all $ u\in V(H_1)$
    and thus, $p_1(K) = H_1\simeq H$. Note that the choice of $H_1$ as ``basis''
    is arbitrary, i.e., we could have started with every subgraph $H_j\subseteq G_j$
    with $H_j\simeq H$ to establish such a map $\varphi_i$ via $f_i\circ f_{j}^{-1}$
    which implies that $p_j(K) = H_j\simeq H$. 
    
    We show now that $K$ is a clique in $G$. Assume, for
    contradiction, that $K$ is not complete. Hence, there are vertices $a = (a_1,\dots,a_n), b=(b_1,
    \dots, b_n)\in V(K)$ such that $\{a, b\}\notin E(K)\subseteq E(G)$. By definition, there are
    distinct indices $i,j \in \{1,\dots,n\}$ such that $\{a_i,b_i\}\in E(G_i)$ but
    $\{a_j,b_j\}\notin E(G_j)$. Note that $\{\varphi_i(a_1),\varphi_i(b_1)\} = \{a_i,b_i\}$ and
    $\{\varphi_j(a_1),\varphi_j(b_1)\} = \{a_j,b_j\}$. Hence, $\{\varphi_i(a_1),\varphi_i(b_1)\}\in
    E(G_i)$ implies $\{a_1,b_1\}\in E(G_1)$ while $\{\varphi_j(a_1),\varphi_j(b_1)\}\notin E(G_j)$
    implies $\{a_1,b_1\}\notin E(G_1)$; a contradiction. Thus, $K$ is a complete subgraph in $G$.
    Assume now, for contradiction, that $K$ is not a maximal clique. Hence, there is a complete
    subgraph $\tilde K$ in $G$ such that $K\subsetneq \tilde K$. Let $x\in V(\tilde K)\setminus
    V(K)$ and consider the graph $K'$ that is induced by the vertices in $V(K)\cup \{x\}]$. By definition, $x$ is adjacent to
    each vertex $w$ in $K$. Hence, either $p_i(x)$ is adjacent to $p_i(w)$ for all $i\in
    \{1,\dots,n\}$ or $p_i(x)$ is not adjacent to $p_i(w)$ for all $i\in \{1,\dots,n\}$. It is now a
    straightforward task to verify that $V(H_i)\cup \{p_i(x)\}$ induces a subgraph
    $H'_i$ in $G_i$ such that $H_i\subsetneq H'_i$ and such that $H'_i\simeq H'_j$; a contradiction
    to the maximality of $H$. Consequently, $K$ must be a maximal clique in $G$. 
\end{proof}

To establish Proposition~\ref{prop:main} we need the following result. 
\begin{lemma}
\label{lem:common-2}
    If $K$ is a clique  in $\star_{i=1}^n G_i$,
	then the graph $H$ induced by the vertex set $\{p_i(w)\mid w\in V(K)\}$ in $G_i$  
 	is a common induced subgraph of $G_1,\dots,G_n$ of size $|V(H)|=|V(K)|$ for every $i\in\{1,\dots,n\}$.
\end{lemma}
\begin{proof}
	Suppose  that there is a clique $K$ in $G\coloneqq \star_{i=1}^n G_i$. Since $K$ is a
	complete graph, definition of the $\star$-product implies that distinct vertices $\tilde u =
	(\tilde u_1,\dots, \tilde u_n)$ and $\tilde w = (\tilde w_1,\dots, \tilde w_n)$ in $K$ differ in
	each of their $i$-th coordinates, i.e., $\tilde w_i\neq \tilde u_i$ for all $i\in \{1,\dots,n\}$. Thus,
	we may assume w.l.o.g.\ that each vertex $\tilde v\in V(K)$ has coordinates $\tilde v = (v,\dots,)$; just
	relabel the vertices in $V(G_i)$, $1\leq i\leq n$, if needed. Hence, $\{\tilde
	w,\tilde u\}\in E(K)$ implies that either $\{w,u\} \in E(G_i)$ for all $i\in \{1,\dots,n\}$ or
	$\{w,u\} \notin E(G_i)$ for all $i\in \{1,\dots,n\}$. Thus, the subgraph $H_i$ incuded by
	$\{w\mid \tilde w\in V(K)\}$ in each $G_i$ satisfy $H_i\simeq H_j$ for all $i,j\in
	\{1,\dots,n\}$. Thus, $H\coloneqq H_1$ is a common subgraph of $G_1,\dots,G_n$. Since distinct
	vertices $\tilde u$ and $\tilde w$ in $K$ differ in their $i$-th coordinates, we have $|V(K)| =
	|V(H)|$. The latter arguments, in particular, imply that a common subgraph $H$ of $G_1,\dots,G_n$
	is isomorphic to the graph induced by the vertex set $\{p_i(w)\mid w\in V(K)\}$ in $G_i$ for
	every $i\in\{1,\dots,n\}$. 
\end{proof}

We may note that maximal cliques in $\star_{i=1}^n G_i$ may not ``correspond'' to maximal common
induced subgraphs of $G_1,\dots,G_n$ (see Fig.~\ref{fig:enter-label}). Nevertheless,
the result holds for maximum common induced subgraphs. 
\begin{figure}
    \centering
    \includegraphics[width=0.5\linewidth]{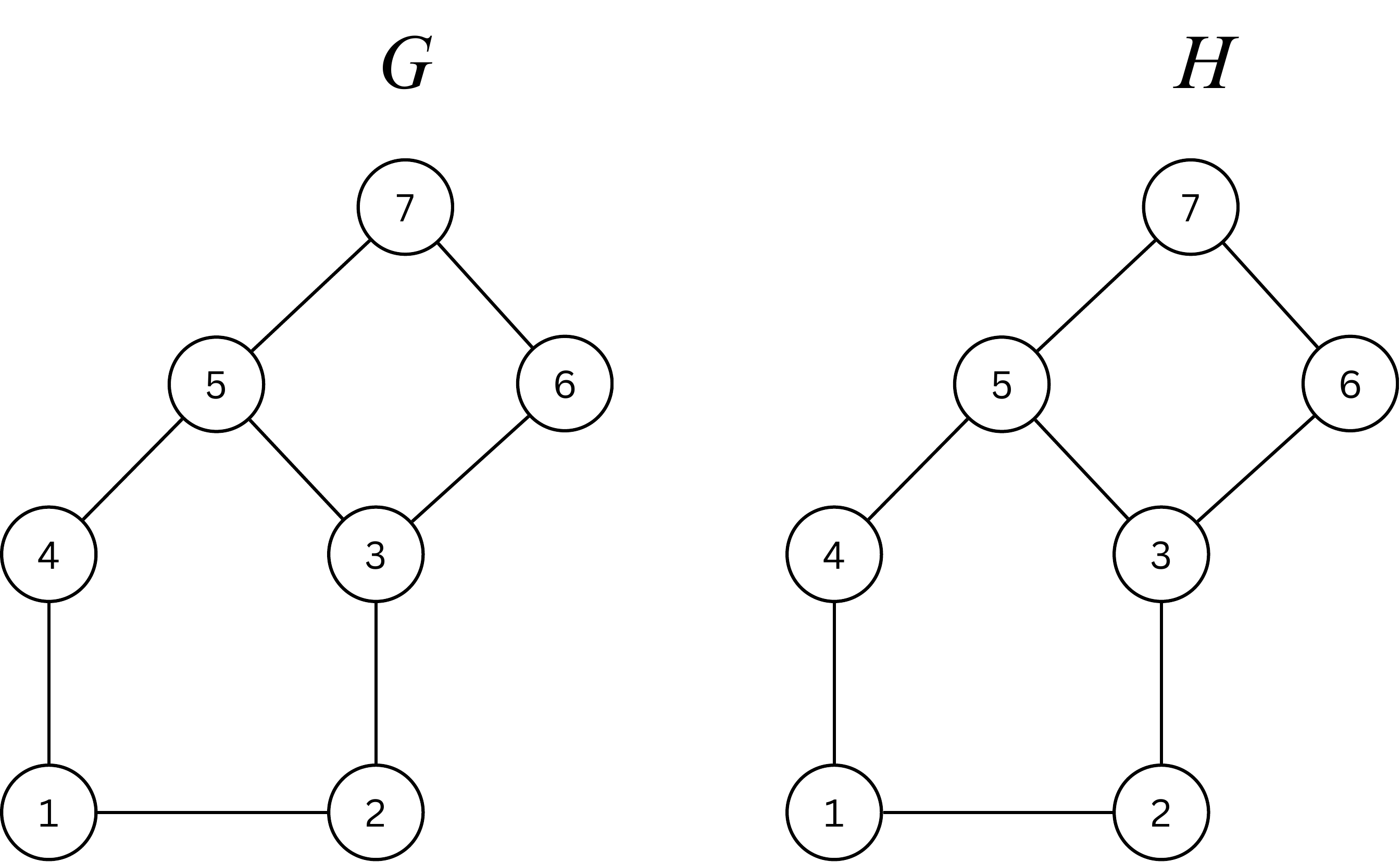}
    \caption{By way of example, The product contains the clique $K = \{(1,5),(2,3), (4,4),(5,1),(3,2)\}\subseteq V(G \star H)$. 
    This clique $K$ is inclusion-maximal. To see this, observe that any additional vertex $(x,y)$ we may add to $K$
    to obtain a larger clique must satisfy $x,y\notin \{1,2,3,4,5\}$ as neither $G$ nor $H$ contains edges $(1,1), \dots, (5,5)$. 
    Hence, only $x,y\in \{6,7\}$ is possible. Since $5$ is adjacent to $7$ and $3$ is adjacent to $6$  but $1$ is not adjacent to 
    $7$ and $6$ in $G$ and $H$, it follows that the subgraph induced by $V(K)\cup \{(x,y)\}$ with $x,y\in \{6,7\}$ cannot form
    a clique in $G \star H$. Since $K$ corresponds to the subgraph induced by $\{1,2,3,4,5\}$ and since $G\simeq H$, 
    it follows that $K$ does not correspond to a maximal common subgraph of $G$ and $H$.  }
    \label{fig:enter-label}
\end{figure}

\primeProp*
\begin{proof}
 	Let $G_1,\dots, G_n$ be graphs and let $\mathcal{K}$ be the set of all maximal cliques $K$ in
 	$\star_{i=1}^n G_i$. Suppose that $H$ is an MVCS of $G_1,\dots,G_n$. Since $H$ has maximum
 	size, it is, in particular, also a maximal common induced subgraph of $G_1,\dots,G_n$. By
 	Lemma~\ref{lem:common-1}, there exists a maximal clique $K$ in $G\coloneqq \star_{i=1}^n G_i$
 	of size $|V(K)| = |V(H)|$. Hence, $K\in \mathcal{K}$. Suppose that $K$ is not a maximum clique
 	within $\mathcal{K}$. Hence, there is a clique $K'\in \mathcal{K}$ of size $|K'|> |K|$. Again, by
 	Lemma~\ref{lem:common-2}, there exists a common subgraph $H'$ of $G_1,\dots,G_n$ of size $|V(H')|
 	=|V(K')|>|V(K)|$; which contradicts the fact that $H$ is an MVCS of $G_1,\dots,G_n$. 
	
	By similar arguments, if $K$ is a maximum clique in $\mathcal{K}$ but the subgraph $H$ of
	$G_1,\dots,G_n$ induced by the vertex set $\{p_1(w)\mid w\in V(K)\}$ in $G_1$ is not maximum, we
	can find a larger common subgraph $H'$ of $G_1,\dots,G_n$ which results in a clique $K'$ in $G$
	of size $|V(K')|>|V(K)|$. Hence, $K'\in \mathcal{K}$ must hold; a contradiction to the choice of
	$K$.  
\end{proof}

 \section{Edge- and vertex-labeled graphs} 

We are mainly interested in molecular graphs, i.e., simple undirected graphs 
whose vertices correspond to different atom types and where edges 
represent pairs of vertices that are connected via a particular type of bond, e.g., hydrogen bonds.  

To cover this, we circumvent the explicit construction of such graphs and instead 
consider a set of vertices $V(G)$, a map $\nu_G \colon V \to \vertexL$ 
assigning to each vertex $v \in V(G)$ a symbol $\nu_G(v) \in \vertexL$ (the type of atom it represents), 
and a map $\varepsilon_G\colon \binom{V}{2} \to \edgeL$ that assigns to each 
2-element subset $\{u,v\} \in \binom{V}{2}$ a symbol $\varepsilon_G(u,v) \in \edgeL$, 
which specifies whether $u$ and $v$ are linked by some type of bond or if 
$u$ and $v$ are not linked at all.  
In particular, we use the special symbol $\NoEdge$ to specify that $\varepsilon_G(u,v)$ 
indicates no bond between $u$ and $v$.  
We write $G = (V(G), \nu_G, \varepsilon_G)$ and call such structures labeled graphs.  
In essence, one can think of such labeled graphs as ``complete edge- and vertex-colored graphs.''

Two labeled graphs $G$ and $H$ are isomorphic, in symbols $G\simeq H$, if there is
a bijective map $\varphi\colon V(G)\to V(H)$ such that the following conditions
are satisfied.
\begin{enumerate}
    \item $\nu_G(v) = \nu_H(\varphi(v))$ for all $v\in V(G)$; and 
    \item $\{u,v\}\in \binom{V(G)}{2}$ if and only if $\{\varphi(u),\varphi(v)\}\in \binom{V(H)}{2}$; and
    \item $\varepsilon_G(\{u,v\}) = \varepsilon_H(\{\varphi(u),\varphi(v)\})$  
          for all  $\{u,v\}\in \binom{V(G)}{2}$.
\end{enumerate}

An \emph{(induced) subgraph $H\subseteq G$} of labeled graph $G$ is a labeled graph $H$ such that 
$V(H)\subseteq V(G)$ and $\nu_H(v) = \nu_G(v)$ for all $v\in V(H)$ and 
$\varepsilon_H(\{u,v\}) = \varepsilon_G(\{u,v\})$  for all  $\{u,v\}\in \binom{V(H)}{2}$.
A \emph{common subgraph} of labeled graphs ${G_1}, \dots, {G_\ell}$ is a 
labeled graph ${H}$ that is isomorphic to an induced subgraph ${H_i}$ in each ${G_i}$. 

\section{Related works}
In this section, we provide an overview of the problem and closely related topics in the literature. The maximum common subgraph (MCS) problem has long been central to molecular informatics, as it provides a rigorous means for comparing molecular structures \cite{irwin2005zinc,geer2010ncbi}. Due to its wide range of applications, several variants of the MCS problem have been investigated in the literature. Early work in chemical database management highlighted the need for robust molecular representations and efficient subgraph matching techniques, leading researchers to adopt graph-based models that capture both atom connectivity and functional group organization \cite{xue2000evaluation,james1995daylight,durant2002reoptimization,takahashi1992automatic,rarey1998feature,harper2004reduced}. A key insight from these studies was the dual nature of the MCS problem: while many investigations have focused on pairwise comparisons between two molecular graphs, there is also significant interest in extending these methods to handle more than two molecules simultaneously in order to extract common structural motifs \cite{hassan2006cheminformatics}.

A further complication in the MCS domain is the distinction between two related formulations. Many methods target the maximal common induced subgraph (MCIS), defined as the largest vertex-induced subgraph common to the compared molecules, while other methods focus on the maximal common edge subgraph (MCES), which maximizes the number of shared edges \cite{Barrow1976-supp,brint1987algorithms}. This duality has led to a generalized classification framework for MCS algorithms that not only categorizes them as exact or approximate but also distinguishes whether the resulting subgraph must be connected or may be disconnected \cite{Bron-Kerbosch-supp,carraghan1990exact,raymond2002rascal,raymond2002heuristics}.

Several early approaches recast the MCS problem as a maximal clique detection problem in the modular product or so-called compatibility (or association) graph. In these methods, potential correspondences between atoms or bonds in two molecular graphs are represented as vertices in a compatibility graph, and a maximal clique corresponds directly to an MCIS (or MCES) of the original graphs. Pioneering clique-detection algorithms such as those by Bron and Kerbosch \cite{Bron-Kerbosch-supp} and subsequent adaptations like the method of Carraghan and Pardalos \cite{carraghan1990exact} have been widely applied. Branch-and-bound procedures, exemplified by the RASCAL algorithm \cite{raymond2002rascal,raymond2002heuristics}, further improve efficiency on dense, chemically labeled graphs.

Alternative strategies include backtracking methods that incrementally build subgraph mappings, as initially proposed by McGregor \cite{mcgregor1982backtrack} and Ullmann \cite{ullmann1976algorithm}. Recent enhancements by Cao et al. \cite{cao2008maximum} and Berlo et al. \cite{van2013efficient} integrate advanced pruning and vertex-ordering heuristics to reduce the search space, thus supporting the computation of both MCIS and MCES solutions. Additionally, dynamic programming (DP) techniques have been successfully applied for special classes of molecular graphs (e.g., trees or outerplanar graphs)\cite{cormen2022introduction,gupta1998finding,boulicaut2008proceedings}, enabling polynomial-time performance where the general MCS problem remains NP-hard.

Extending MCS computation from pairwise comparisons to multiple molecules is nontrivial and introduces additional challenges. Simple strategies based on the aggregation of pairwise MCS results often fail to capture the true common substructure shared by all molecules \cite{hassan2006cheminformatics}. Instead, approaches have been developed to enumerate all maximal common substructures across the set, although these typically sacrifice some of the optimization benefits available in the two-molecule case.

Cardone and Quer proposed a series of heuristics, both sequential and parallel (including GPU-based implementations), to approximate the solution when the exact method becomes computationally prohibitive. In addition, they analyze various sorting heuristics to order both the vertices and the graphs, achieving substantial speed ups in practice \cite{cardone2023multi}. 
Englert and Kovacs in \cite{englert2015efficient} contribute to the heuristic side of MCS search by developing some methods that balance speed with the quality of the resulting atom mappings. Their work focuses on improving both the efficiency and the chemical relevance of the mappings obtained, an aspect particularly important for applications such as molecule alignment. Their heuristics provide an alternative to more computationally expensive exact methods, albeit with some loss in optimality .
Dalke and Hastings presented FMCS, an algorithm for the multiple MCS problem based on subgraph enumeration and subgraph isomorphism testing. FMCS approaches the problem by recursively reducing the multi-graph case to successive pairwise searches, yet it integrates additional heuristics to handle challenges such as chemistry perception and timeout errors \cite{dalke2013fmcs}. 

A notable contribution in the multiple-molecule setting is the MultiMCS algorithm by Hariharan et al. \cite{MultiMCS-supp}, which targets connected MCESs across several molecules. While MultiMCS represents a valuable step toward extending common substructure search to larger datasets, several issues limit its practical impact. First, MultiMCS is presented without a formal proof of exactness—its correctness is merely claimed, whereas our algorithm includes a rigorous proof ensuring optimality. Second, the absence of any publicly available implementation in MultiMCS restricts reproducibility and meaningful benchmarking; in contrast, our work is accompanied by an open-source GitHub repository that enables transparent comparison and widespread adoption. Third, the description of their overall algorithm (Algorithm~1) is problematic: when both $X$ and $Y$ are connected subgraphs of a molecule, $M_1$, their intersection ($Y\cap X$) may easily be disconnected. This oversight implies that the output might not fulfill the connectedness criterion as claimed, a critical issue that our approach carefully addresses. These shortcomings underscore the necessity for both theoretical validation and practical accessibility, aspects that our work robustly provides.

Finally, the literature exhibits a rich variety of approaches to the MCS problem. While many early methods focused on exact solutions for two molecular graphs using either MCIS or MCES formulations, more recent work has diversified to include approximate methods, DP for special graph classes, and algorithms that extend to multiple molecules. The choice of method, whether it is an exact, connected MCIS algorithm such as that by McGregor \cite{mcgregor1982backtrack} or a disconnected MCES approach like RASCAL \cite{raymond2002rascal,raymond2002heuristics}, remains highly application-specific and depends on factors such as graph sparsity, the extent of chemical labeling, and the intended scope of similarity assessment. Continued efforts to establish standardized benchmarking \cite{trinajstic2018chemical,gund1977three} and to develop publicly available implementations are critical to further advancing this field \cite{whitneycongruent,gund1979pharmacophoric}.

\end{appendix}

\section*{Acknowledgements}
This work was funded by the Novo Nordisk Foundation as part of the project MATOMIC
(Mathematical Modelling for Microbial Community Induced Metabolic Diseases, 0066551) and from the European Unions Horizon Europe
Doctoral Network programme under the Marie-Sklodowska-Curie grant agreement 
No 101072930 (TACsy – Training Alliance for Computational systems chemistry).

\bibliographystyle{unsrt}  
\bibliography{references}  

\end{document}